\documentclass[reqno,11pt]{amsart}
\DeclareMathOperator*{\esssup}{\mathrm{ess\,sup}}

\usepackage{amsmath,amssymb}
\usepackage{dsfont}         % \mathds{R} np.

\newtheorem{theorem}{Theorem}[section]
\newtheorem{corollary}[theorem]{Corollary}
\newtheorem{lemma}[theorem]{Lemma}
\newtheorem{proposition}[theorem]{Proposition}
\theoremstyle{definition}
\newtheorem{definition}[theorem]{Definition}
\theoremstyle{remark}
\newtheorem{remark}[theorem]{Remark}
\numberwithin{equation}{section}

\title[States of a Continuum Kawasaki Model]{The Global Evolution of States of a Continuum Kawasaki Model with Repulsion}

\author{Joanna Bara{\'n}ska}

\address{Instytut Matematyki, Uniwersytet Marii Curie-Sk{\l}odowskiej, 20-031 Lublin, Poland}
\email{asia13p@wp.pl}

\author{ Yuri  Kozitsky}

\address{Instytut Matematyki, Uniwersytet Marii Curie-Sk{\l}odowskiej, 20-031 Lublin, Poland}
\email{jkozi@hektor.umcs.lublin.pl}

\begin{document}



\begin{abstract}
An infinite system of point particles performing random jumps in
$\mathds{R}^d$  with repulsion is studied. The states of the system
are probability measures on the space of particle's configurations.
The result of the paper is the construction of the global in time
evolution of states with the help of the corresponding correlation
functions. It is proved that for each initial sub-Poissonian state
$\mu_0$, the constructed evolution $\mu_0 \mapsto \mu_t$ preserves
this property. That is, $\mu_t$ is sub-Poissonian for all $t>0$.
\end{abstract}
\maketitle

%\tableofcontents

\section{Introduction}
\label{S1}

\subsection{Posing}

In this paper, we continue dealing with the Kawasaki model studied
in \cite{BKKK}. The model describes the evolution of an infinite
system of point particles placed in $\mathds{R}^d$ which perform
random jumps with repulsion. The phase space of the model is the set
$\Gamma$ of all subsets $\gamma \subset \mathds{R}^d$ such that the
set $\gamma\cap\Lambda$ is finite whenever $\Lambda \subset
\mathds{R}^d$ is compact. This set is equipped with a measurability
structure that allows for considering the probability measures on
$\Gamma$ as states of the system. Among them one may distinguish
Poissonian states in which the particles are independently
distributed over $\mathds{R}^d$. In \emph{sub-Poissonian} states,
the dependence between the particle's positions  is not too strong
(see the next subsection). In \cite{BKKK}, the evolution $\mu_0
\mapsto \mu_t$ of the system's states was shown to hold in the set
of sub-Poissonian states for $t< T$ with some $T < \infty$. The main
result of the present study consists in proving the existence of
such an evolution for all $t>0$. This is the first result of this
kind for infinite continuum systems of point particles performing
jumps with interaction. The case of free jumps was described in
\cite{BK,KLR}.

As was shown in \cite{KK}, for infinite particle systems with
birth-and-death dynamics the states remain sub-Poissonian globally
in time if the birth of the particles is in a sense controlled by
their death. For conservative dynamics in which the particles just
change their positions, the interaction may in general change the
sub-Poissonian character of the state in finite time (even cause an
explosion), e.g., due to an infinite number of simultaneous
correlated jumps.   Thus, the conceptual outcome of the present
study is that this is not the case for the considered model. The
important peculiarity of this result is that it has been obtained by
methods different from those used in \cite{KK}. We believe that a
combination of these methods with those of \cite{KK} can be of great
use in studying evolution of systems in which birth-and-death
processes are accompanied by random motion, e.g., individual-based
models of disease spread.

\subsection{Presenting the result}

To characterize states of an infinite particle system one employs
{\it observables} -- suitable functions $F:\Gamma \to \mathds{R}$.
Their evolution is described by the Kolmogorov equation
\begin{equation}
  \label{1}
 \frac{d}{dt} F_t = L F_t, \qquad F_t|_{t=0} = F_0,
\end{equation}
where the operator $L$ specifies the model. In our case, it has the
following form
\begin{eqnarray}\label{LF}
(LF)(\gamma) &= &\sum_{x\in \gamma}\int_{\mathds{R}^d} c(x,y,\gamma)
[F(\gamma\backslash x \cup y)-F(\gamma)]d y,
\end{eqnarray}
with $c$ given in (\ref{c}) below. The evolution of states is
supposed to be derived from the Fokker-Planck equation
\begin{equation}
  \label{1a}
\frac{d}{dt} \mu_t = L^* \mu_t, \qquad \mu_t|_{t=0} = \mu_0,
\end{equation}
related to that in (\ref{1}) by the duality
\begin{equation}
  \label{1b}
  \int_{\Gamma}F_t (\gamma) \mu_0 ( d \gamma) = \int_{\Gamma}F_0 (\gamma) \mu_t ( d
  \gamma).
\end{equation}
As is usual for models of this kind, the direct meaning of (\ref{1})
or (\ref{1a}) can only be given for states of finite systems, cf.
\cite{K}. In this case, the Banach space where the Cauchy problem in
(\ref{1a}) is defined can be the space of signed measures with
finite variation.

In this work, we continue following the approach in which the
evolution of states is described without the direct use of
(\ref{1a}), see \cite{BKKK,FKKK,KK} and the references therein. To
explain its essence let us consider the set of all compactly
supported continuous functions $\theta:\mathbb{R}^d\to (-1,0]$. For
a state $\mu$, its {\it Bogoliubov} functional \cite{FKO} is defined
as
\begin{equation}
  \label{I1}
B_\mu (\theta) = \int_{\Gamma} \prod_{x\in \gamma} ( 1 + \theta (x))
\mu( d \gamma),
\end{equation}
with $\theta$ running through the mentioned set of functions. For
the homogeneous Poisson measure $\pi_\varkappa$, $\varkappa>0$, the
functional (\ref{I1}) takes the form
\begin{equation*}
  %\label{I2}
B_{\pi_\varkappa} (\theta) = \exp\left(\varkappa
\int_{\mathbb{R}^d}\theta (x) d x \right).
\end{equation*}
In state $\pi_\varkappa$, the particles are independently
distributed over $\mathbb{R}^d$ with density $\varkappa$. The set of
{\it sub-Poissonian} states $\mathcal{P}_{\rm exp}(\Gamma)$ is then
defined as that containing all those states $\mu$ for which $B_\mu$
can be continued, as a function of $\theta$, to an exponential type
entire function on $L^1 (\mathbb{R}^d)$. This exactly means that
$B_\mu$ can be written in the form
\begin{eqnarray}
  \label{I3}
B_\mu(\theta) = 1+ \sum_{n=1}^\infty
\frac{1}{n!}\int_{(\mathbb{R}^d)^n} k_\mu^{(n)} (x_1 , \dots , x_n)
\theta (x_1) \cdots \theta (x_n) d x_1 \cdots d x_n,
\end{eqnarray}
where $k_\mu^{(n)}$ is the $n$-th order correlation function of the
state $\mu$. It is a symmetric element of $L^\infty
((\mathbb{R}^d)^n)$ for which
\begin{equation}
\label{I4}
  \|k^{(n)}_\mu \|_{L^\infty
((\mathbb{R}^d)^n)} \leq C \exp( \vartheta n), \qquad n\in
\mathbb{N}_0,
\end{equation}
with some $C>0$ and $\vartheta \in \mathbb{R}$. Note that
$k_{\pi_\varkappa}^{(n)} (x_1 , \dots , x_n)= \varkappa^n$. Note
also that (\ref{I3}) can be viewed as an analog of the Taylor
expansion of the characteristic function of a probability measure.
That is why, $k^{(n)}_\mu$ are also called \emph{moment functions}.

Under standard conditions imposed on the jump kernel $c$, see
(\ref{c}) -- (\ref{33}), we prove that the correlation functions
evolve $k_{\mu_0}^{(n)}\mapsto k_t^{(n)}$ in such a way that each
$k_t^{(n)}$, $t>0$, is the correlation function of a unique
sub-Poissonian measure $\mu_t$, see Theorem \ref{1tm}. Moreover,
assuming that $k_{\mu_0}^{(n)}$ satisfies (\ref{I4}), we show that
the following holds
\[
\forall t >0 \quad  \forall n\in \mathbb{N}_0 \qquad 0\leq
k^{(n)}_{t}(x_1 , \dots , x_n) \leq C \exp\left( n[\vartheta +
\alpha t]\right),
\]
where $\alpha >0$ is a model parameter, see (\ref{10}).

\section{Preliminaries and the Model}

Here we briefly present necessary information on the subject -- its
more detailed description can be found in \cite{BKKK,FKKK,FKO,KK}
and in the literature quoted in these works.

\subsection{Configuration spaces}

Let $\mathcal{B}(\mathds{R}^d)$ and $\mathcal{B}_{\rm
b}(\mathds{R}^d)$ denote the sets of all Borel and all bounded Borel
subsets of $\mathds{R}^d$, respectively. The configuration space
$\Gamma$ mentioned above is equipped with the vague topology and
thus with the corresponding Borel $\sigma$-field
$\mathcal{B}(\Gamma)$.  For $\Lambda\in \mathcal{B}(\mathds{R}^d)$,
we set
\[
\Gamma_\Lambda = \{\gamma\in \Gamma: \gamma \subset \Lambda \}.
\]
Clearly $\Gamma_\Lambda \in \mathcal{B}(\Gamma)$, and hence
\[
 \mathcal{B}(\Gamma_\Lambda):=\{ A \cap \Gamma_\Lambda : A \in \mathcal{B}(\Gamma)\}
\]
is a sub-field of $\mathcal{B}(\Gamma)$. Let $p_{\Lambda}:\Gamma\to
\Gamma_\Lambda$ be the projection $p_\Lambda (\gamma) =
\gamma_\Lambda=\gamma \cap \Lambda$. It is clearly measurable, and
thus the sets
\begin{equation*}
% \label{4}
 p^{-1}_\Lambda(A_\Lambda) :=\{ \gamma\in \Gamma: p_\Lambda (\gamma) \in A_\Lambda \},
 \quad A_\Lambda \in \mathcal{B}(\Gamma_\Lambda)
\end{equation*}
belong to $\mathcal{B}(\Gamma)$ for each Borel $\Lambda$. Let
$\mathcal{P}(\Gamma)$ denote the set of all  probability measures on
$(\Gamma, \mathcal{B}(\Gamma))$. For a given $\mu\in
\mathcal{P}(\Gamma)$, its projection on $(\Gamma_\Lambda,
\mathcal{B} (\Gamma_\Lambda))$ is defined as
\begin{equation}
 \label{5}
\mu^\Lambda (A_\Lambda) = \mu\left(p^{-1}_\Lambda (A_\Lambda)
\right), \qquad A_\Lambda \in \mathcal{B}(\Gamma_\Lambda).
\end{equation}
Let $\Gamma_0$ be the set of all finite $\gamma \in \Gamma$. It is
an element of $\mathcal{B}(\Gamma)$ as each of $\gamma \in \Gamma_0$
belongs to a certain $\Gamma_\Lambda$, $\Lambda \in \mathcal{B}_{\rm
b}(\mathds{R}^d)$. Note that $\Gamma_\Lambda \subset \Gamma_0$ for
each such $\Lambda$. Set $\mathds{N}_0 = \mathds{N}\cup \{0\}$. It
can be proved that a function $G:\Gamma_0 \to \mathds{R}$ is
$\mathcal{B}(\Gamma)/\mathcal{B}(\mathds{R} )$-measurable if and
only if, for each $n\in \mathds{N}_0$, there exists a symmetric
Borel function $G^{(n)}: (\mathds{R}^{d})^{n} \to \mathds{R}$ such
that
\begin{equation}
 \label{7}
 G(\eta) = G(\eta) = G^{(n)} ( x_1, \dots , x_{n}),
\end{equation}
for $\eta = \{ x_1, \dots , x_{n}\}$ .
\begin{definition}
  \label{Gdef}
A measurable function $G:\Gamma_0 \to \mathds{R}$ is said have
bounded support if: (a) there exists $\Lambda \in \mathcal{B}_{\rm
b} (\mathds{R}^d)$ such that $G(\eta) = 0$ whenever $\eta\cap
\Lambda^c \neq \emptyset$; (b) there exists $N\in \mathds{N}_0$ such
that $G(\eta)=0$ whenever $|\eta|
>N$. Here $\Lambda^c := \mathds{R}^d
\setminus \Lambda$ and $|\cdot |$ stands for cardinality. By
$\Lambda(G)$ and $N(G)$ we denote the smallest $\Lambda$ and $N$
with the properties just mentioned. By $B_{\rm bs}(\Gamma_0)$ we
denote the set of all such functions.
\end{definition}
The Lebesgue-Poisson measure $\lambda$ on $(\Gamma_0,
\mathcal{B}(\Gamma_0))$ is defined by the following formula
\begin{eqnarray}
\label{8} \int_{\Gamma_0} G(\eta ) \lambda ( d \eta)  = G(\emptyset)
+ \sum_{n=1}^\infty \frac{1}{n! } \int_{(\mathds{R}^d)^{n}} G^{(n)}
( x_1, \dots , x_{n} ) d x_1 \cdots dx_{n},
\end{eqnarray}
which has to hold for all $G\in B_{\rm bs}(\Gamma_0)$. For
$\gamma\in \Gamma$, by writing $\eta \Subset \gamma$ we mean that
$\eta\subset \gamma$ is nonempty and finite. For $G\in B_{\rm
bs}(\Gamma)$, we set
\begin{equation}
  \label{9a}
(KG)(\gamma) = \sum_{\eta \Subset \gamma} G(\eta).
\end{equation}
Note that the sum in (\ref{9a}) is finite and $KG$ is a cylinder
function on $\Gamma$. The latter means that it is
$\mathcal{B}(\Gamma_{\Lambda(G)})$-measurable, see Definition
\ref{Gdef}. Moreover,
\begin{equation}
  \label{9b}
|(KG)(\gamma)| \leq \left( 1 + |\gamma\cap\Lambda (G)|\right)^{N(G)}
.
\end{equation}

\subsection{Correlation functions}

Like in (\ref{7}), we introduce $k_\mu : \Gamma_0 \to \mathds{R}$
such that $k_\mu(\eta) = k^{(n)}_\mu (x_1, \dots , x_n)$ for $\eta =
\{x_1, \dots , x_n\}$, $n\in \mathds{N}$. We also set
$k_\mu(\emptyset)=1$. With the help of the measure introduced in
(\ref{8}), the formulas for $B_\mu$ in (\ref{I1}) and (\ref{I3}) can
be combined into the following formula
\begin{eqnarray}
  \label{1fa}
 B_\mu (\theta)& = & \int_{\Gamma_0} k_\mu(\eta) \prod_{x\in \eta} \theta (x) \lambda (d\eta)=: \int_{\Gamma_0} k_\mu(\eta) e( \eta; \theta) \lambda (d \eta)
 \\[.2cm]
 & = &  \int_{\Gamma} \prod_{x\in \gamma} (1+ \theta (x)) \mu (d \gamma) =: \int_{\Gamma} F_\theta (\gamma) \mu(d
 \gamma). \nonumber
\end{eqnarray}
Thereby, we can transform  the action of $L$ on $F$, as in
(\ref{LF}), to the action of $L^\Delta$ on $k_\mu$ according to the
rule
\begin{equation}
  \label{1g}
\int_{\Gamma}(L F_\theta) (\gamma) \mu(d \gamma) = \int_{\Gamma_0}
(L^\Delta k_\mu) (\eta) e(\eta;\theta)
 \lambda (d \eta).
\end{equation}
This will allow us to pass from (\ref{1}) to the corresponding
Cauchy problem for the correlation functions, cf. (\ref{16}) below.
The main advantage here is that $k_\mu$ is a function of {\em
finite} configurations.

For $\mu \in \mathcal{P}_{\rm exp}(\Gamma)$ and  $\Lambda \in
\mathcal{B}_{\rm b}(\mathds{R}^d)$, let $\mu^\Lambda$ be as in
(\ref{5}). Then $\mu^\Lambda$ is absolutely continuous with respect
to the restriction  $\lambda^\Lambda$ to
$\mathcal{B}(\Gamma_\Lambda)$  of the measure defined in (\ref{8}),
and hence we may write
\begin{equation}
\label{9c} \mu^\Lambda (d \eta ) = R^\Lambda_\mu (\eta)
\lambda^\Lambda ( d \eta), \qquad \eta \in \Gamma_\Lambda.
\end{equation}
Then the correlation function $k_\mu$ and the Radon-Nikodym
derivative $R_\mu^\Lambda$ satisfy
\begin{eqnarray}
  \label{9d}
k_\mu(\eta) & = & \int_{\Gamma_\Lambda} R^\Lambda_\mu (\eta \cup
\xi) \lambda^\Lambda ( d\xi).
\end{eqnarray}
Note that (\ref{9d}) relates $R^\Lambda_\mu$ with the restriction of
$k_\mu$ to $\Gamma_\Lambda$. The fact that these are the
restrictions of one and the same function
$k_\mu:\Gamma_0\to\mathds{R}$ corresponds to the Kolmogorov
consistency of the family $\{\mu^\Lambda: \Lambda \in
\mathcal{B}(\mathds{R}^d)\}$.

By (\ref{9a}), (\ref{5}), and (\ref{9c}) we get
\begin{equation}
  \label{9e}
\int_{\Gamma} (KG)(\gamma) \mu(d\gamma) = \langle \! \langle G,
k_\mu \rangle \!\rangle,
\end{equation}
holding for each $G\in B_{\rm bs}(\Gamma_0)$ and $\mu \in
\mathcal{P}_{\rm exp}(\Gamma)$. Here
\begin{equation}
  \label{9f}
\langle \! \langle G, k \rangle \!\rangle := \int_{\Gamma_0} G(\eta)
k(\eta) \lambda (d \eta),
\end{equation}
for suitable $G$ and $k$. Define
\begin{equation}
  \label{9g}
B^\star_{\rm bs} (\Gamma_0) =\{ G\in B_{\rm bs}(\Gamma_0):
(KG)(\gamma) \geq 0 \ {\rm for} \ {\rm all} \ \gamma\in \Gamma\}.
\end{equation}
By \cite[Theorems 6.1 and 6.2 and Remark 6.3]{Tobi} one can prove
the next statement.
\begin{proposition}
  \label{Gpn}
Let  a measurable function $k : \Gamma_0 \to \mathds{R}$  have the
following properties:
\begin{eqnarray}
  \label{9h}
& (a) & \ \langle \! \langle G, k \rangle \!\rangle \geq 0, \qquad
{\rm for} \ {\rm all} \ G\in B^\star_{\rm bs} (\Gamma_0);\\[.2cm]
& (b) & \ k(\emptyset) = 1; \qquad (c) \ \ k(\eta) \leq
 C^{|\eta|} ,
\nonumber
\end{eqnarray}
with (c) holding for some $C >0$ and $\lambda$-almost all $\eta\in
\Gamma_0$. Then there exists a unique $\mu \in \mathcal{P}_{\rm
exp}(\Gamma)$ for which $k$ is the correlation function.
\end{proposition}

\subsection{The model}

The model we consider is specified by the operator $L$ given in
(\ref{LF}) with
\begin{equation}\label{c}
c(x,y,\gamma)  =  a(x-y)\exp\left(-\sum_{z \in
\gamma}\phi(y-z)\right).
\end{equation}
The jump kernel $a: \mathds{R}^d \rightarrow [0,+\infty)$ is such
that $a(x)=a(-x)$ and
\begin{equation}
 \label{10}
\int_{\mathds{R}^d}a(x)d x =: \alpha < \infty,
\end{equation}
whereas the repulsion potential $\phi: \mathds{R}^d \rightarrow
[0,+\infty)$, $\phi(x) = \phi(-x)$, is supposed to be such that
\begin{eqnarray}
  \label{33}
\int_{\mathds{R}^d} \phi(x) d x =:\langle \phi \rangle < \infty,
\qquad \esssup_{x\in \mathds{R}^d} \phi (x) =:\bar{\phi} < \infty.
\end{eqnarray}
Then also
\begin{equation}
 \label{11}
 \int_{\mathds{R}^d}\bigg{(}1-\exp(-\phi(x))\bigg{)}d x \leq \langle \phi \rangle.
\end{equation}
By (\ref{LF}) and (\ref{1g}) one obtains, cf. \cite[Eq.
(3.1)]{BKKK},
\begin{eqnarray}
  \label{15}
(L^\Delta k) (\eta) & = & \sum_{y\in \eta} \int_{\mathds{R}^d} a
(x-y) e(\tau_y;\eta\setminus y \cup x) (Q_y k) (\eta\setminus y \cup x) d x \nonumber \\[.2cm]
& - & \sum_{x\in \eta} \int_{\mathds{R}^d} a (x-y) e(\tau_y;\eta)
(Q_y k) (\eta) d y.
\end{eqnarray}
Here $e$ is as in (\ref{1fa}),
\begin{equation}
  \label{13}
(Q_y k) (\eta)  :=  \int_{\Gamma_0} k(\eta \cup \xi)e(t_y ;\xi)
\lambda(d\xi),
\end{equation}
and
\begin{equation}
  \label{12}
 \tau_x (y) := \exp(-\phi (x-y)), \quad  t_x (y):= \tau_x (y) - 1,
 \quad \
x, y\in \mathds{R}^d.
\end{equation}

\section{The result}

 \label{S3}

As mentioned above, instead of directly dealing with the problem in
(\ref{1a}) we pass from $\mu_0$ to the corresponding correlation
function $k_{\mu_0}$ and then consider the problem
\begin{equation}
  \label{16}
\frac{d}{dt} k_t = L^\Delta k_t, \qquad k_t|_{t=0} = k_{\mu_0}
\end{equation}
with $L^\Delta$ given in (\ref{15}). The aim is to prove the
existence of a unique global solution $k_t$ of (\ref{16}) which is
the correlation function of a unique state $\mu_t \in
\mathcal{P}_{\rm exp}(\Gamma)$.

We begin by defining  (\ref{16}) in the corresponding spaces of
functions $k:\Gamma_0 \to \mathds{R}$. From the very representation
(\ref{I3}), see also (\ref{1fa}), it follows that $\mu \in
\mathcal{P}_{\rm exp}(\Gamma)$ implies
\begin{equation*}
 % \label{17}
 |k_\mu (\eta)| \leq C \exp( \vartheta  |\eta|),
\end{equation*}
holding for $\lambda$-almost all $\eta\in \Gamma_0$, some $C>0$, and
$\vartheta\in \mathds{R}$. Keeping this in mind we set
\begin{equation}
  \label{17a}
 \|k\|_\vartheta = \esssup_{\eta \in \Gamma_0}\left\{ |k_\mu (\eta)| \exp\big{(} - \vartheta
  |\eta| \big{)} \right\}.
\end{equation}
Then
\begin{equation*}
  %\label{18}
\mathcal{K}_\vartheta := \{ k:\Gamma_0\to \mathds{R}:
\|k\|_\vartheta <\infty\}
\end{equation*}
is a Banach space with norm (\ref{17a}) and the usual linear
operations. In fact, we are going to use the ascending scale of such
spaces $\mathcal{K}_\vartheta$, $\vartheta \in \mathds{R}$, with the
property
\begin{equation}
  \label{19}
\mathcal{K}_\vartheta \hookrightarrow \mathcal{K}_{\vartheta'},
\qquad \vartheta < \vartheta',
\end{equation}
where $\hookrightarrow$  denotes continuous embedding. Set, cf.
(\ref{9e}), (\ref{9f}), and (\ref{9g}),
\begin{equation}
  \label{19a}
\mathcal{K}^\star_\vartheta =\{k\in \mathcal{K}_\vartheta: \langle
\! \langle G,k \rangle \! \rangle \geq 0 \ {\rm for} \ {\rm all} \
G\in B^\star_{\rm bs} (\Gamma_0)\},
\end{equation}
which is a subset of the cone
\begin{equation}
  \label{19b}
\mathcal{K}^+_\vartheta =\{k\in \mathcal{K}_\vartheta: k(\eta) \geq
0 \ \ {\rm for} \  \lambda-{\rm almost} \ {\rm all} \ \eta \in
\Gamma_0\}.
\end{equation}
By Proposition \ref{Gpn} it follows that each $k\in
\mathcal{K}^\star_\vartheta$ such that $k(\emptyset) = 1$ is the
correlation function of a unique state $\mu\in \mathcal{P}_{\rm
exp}(\Gamma)$. Then we define
\begin{equation}
  \label{19c}
\mathcal{K} = \bigcup_{\vartheta \in \mathds{R}}
\mathcal{K}_\vartheta, \qquad \mathcal{K}^\star = \bigcup_{\vartheta
\in \mathds{R}} \mathcal{K}_\vartheta^\star.
\end{equation}
As a sum of Banach spaces, the linear space $\mathcal{K}$ is
equipped with the corresponding inductive topology which turns it
into a locally convex space.

For a given $\vartheta\in \mathds{R}$, by (\ref{15}) -- (\ref{12})
we define $L^\Delta_\vartheta$ as a linear operator in
$\mathcal{K}_\vartheta$ with domain
\begin{equation}
  \label{20}
\mathcal{D} (L^\Delta_\vartheta) = \{ k\in \mathcal{K}_\vartheta:
L^\Delta k \in \mathcal{K}_\vartheta\}.
\end{equation}
\begin{lemma}
  \label{1lm}
For each $\vartheta'' < \vartheta$, cf. (\ref{19}), it follows that
$\mathcal{K}_{\vartheta''} \subset \mathcal{D}
(L^\Delta_\vartheta)$.
\end{lemma}
\begin{proof}
For $\vartheta'' < \vartheta$, by (\ref{11}), (\ref{13}),
(\ref{12}), and (\ref{17a}) we have
\begin{eqnarray*}
  %\label{21}
\left\vert (Q_y k)(\eta)\right\vert   & \leq & \|k\|_{\vartheta''}
\exp\left(  \vartheta'' |\eta|\right) \qquad  \\[.2cm] & \times & \int_{\Gamma_0} \exp\left(  \vartheta'' |\xi|\right) \prod_{z\in \xi}
\bigg{(} 1 - \exp\left( - \phi (z-y)\right)\bigg{)} \lambda (
d\xi)\nonumber \\[.2cm]
& \leq & \|k\|_{\vartheta''} \exp\left(  \vartheta'' |\eta|\right)
\exp\left(\langle \phi \rangle e^{\vartheta''} \right) . \nonumber
\end{eqnarray*}
Now we apply the latter estimate and (\ref{10}) in (\ref{15}) and
obtain
\begin{equation}
  \label{23}
|(L^\Delta k)(\eta)| \leq  2\alpha\|k\|_{\vartheta''} \exp\left(  \vartheta'' |\eta|\right) \\[.2cm]
|\eta| \exp\left(\langle \phi \rangle e^{\vartheta''} \right).
\end{equation}
By means of the inequality $x\exp(-\sigma x) \leq 1/ e \sigma$, $x,
\sigma
>0$, we get from (\ref{17a}) and (\ref{23}) the following estimate
\begin{gather}
  \label{24}
 \|L^\Delta k\|_{\vartheta}  \leq  \frac{2\alpha\|k\|_{\vartheta''}}{e(\vartheta - \vartheta'')}
  \exp\left(
\langle \phi \rangle e^{\vartheta''}\right),
\end{gather}
which yields the proof.
\end{proof}
\begin{corollary}
  \label{Gco}
For each $\vartheta,\vartheta''\in \mathds{R}$ such that
$\vartheta'' < \vartheta$, the expression in (\ref{15}) defines a
bounded linear operator $L^\Delta_{\vartheta\vartheta''}:
\mathcal{K}_{\vartheta''}\to \mathcal{K}_{\vartheta}$ the norm of
which can be estimated by means of (\ref{24}).
\end{corollary}
In what follows, we consider two types of operators defined by the
expression in (\ref{15}): (a) unbounded operators
$(L^\Delta_\vartheta, \mathcal{D}(L^\Delta_\vartheta))$,
$\vartheta\in \mathds{R}$, with domains as in (\ref{20}) and Lemma
\ref{1lm}; (b) bounded operators $L^\Delta_{ \vartheta \vartheta''}$
as in Corollary \ref{Gco}. These operators are related to each other
in the following way:
\begin{equation}
  \label{24a}
\forall \vartheta'' < \vartheta \ \  \forall k \in
\mathcal{K}_{\vartheta''} \qquad L^\Delta_{\vartheta
 \vartheta''} k = L^\Delta_{\vartheta} k.
\end{equation}
 By means of the bounded operators $L^\Delta_{\vartheta
 \vartheta''} : \mathcal{K}_{\vartheta''} \to \mathcal{K}_{\vartheta}$
 we  also
define a continuous linear operator $L^\Delta:\mathcal{K} \to
\mathcal{K} $, see (\ref{19c}). In view of this, we consider the
following two equations. The first one is
\begin{equation}
  \label{24c}
\frac{d}{dt} k_t = L^\Delta_\vartheta k_t, \qquad k_t|_{t=0} =
k_{\mu_0},
\end{equation}
considered as an equation in a given Banach space
$\mathcal{K}_{\vartheta}$. The second equation is (\ref{16}) with
$L^\Delta$ given in (\ref{15}) considered in the locally convex
space $\mathcal{K}$.
\begin{definition}
  \label{S1df}
By a solution of (\ref{24c}) on a time interval, $[0,T)$, $T\leq
+\infty$, we mean a continuous map $[0,T)\ni t \mapsto k_t \in
\mathcal{D} (L^\Delta_\vartheta)$ such that the map $[0,T)\ni t
\mapsto d k_t / dt\in \mathcal{K}_\vartheta$ is also continuous and
both equalities in (\ref{24c}) are satisfied. Likewise, a
continuously differentiable map $[0,T)\ni t \mapsto k_t \in
\mathcal{K}$ is said to be a solution of (\ref{16}) in $\mathcal{K}$
if both equalities therein are satisfied for all $t$. Such a
solution is called global if $T=+\infty$.
\end{definition}
\begin{remark}
  \label{D1rk}
The map $[0,T)\ni t\mapsto k_t \in \mathcal{K}$ is a solution of
(\ref{16}) if and only if, for each $t \in [0, T)$, there exists
$\vartheta''\in \mathds{R}$ such that $k_t\in
\mathcal{K}_{\vartheta''}$ and, for each $\vartheta > \vartheta''$,
the map $t\mapsto k_t$ is continuously differentiable at $t$ in
$\mathcal{K}_\vartheta$ and $d k_t/ dt = L^\Delta_\vartheta k_t =
L^\Delta_{\vartheta \vartheta''} k_t$.
\end{remark}
Our main result is contained in the following statement.
\begin{theorem}
  \label{1tm}
Assume that (\ref{10}) and (\ref{33}) hold. Then for each $\mu_0 \in
\mathcal{P}_{\rm exp}(\Gamma)$, the problem (\ref{16})  with $k_0 =
k_{\mu_0}$  has a unique global solution $k_t \in
\mathcal{K}^\star\subset \mathcal{K}$ which has the property
$k_t(\emptyset) = 1$. Therefore, for each $t\geq 0$ there exists a
unique state $\mu_t\in \mathcal{P}_{\rm exp}(\Gamma)$ such that $k_t
= k_{\mu_t}$. Moreover, let $k_0$ and $C>0$ be such that $k_0(\eta)
\leq C^{|\eta|}$ for $\lambda$-almost all $\eta\in \Gamma_0$, see
(\ref{9h}). Then the mentioned solution satisfies
\begin{equation}
  \label{24d}
 \forall t\geq 0 \qquad 0\leq k_t (\eta) \leq C^{|\eta|} \exp\left( t
 \alpha |\eta| \right) .
\end{equation}
\end{theorem}

%\subsection{Comments}

\section{The Proof of Theorem \ref{1tm}}
\label{SPr}

Our strategy of the proof resembles that used in \cite{KK}.
Basically, it consist in performing the following three steps: (a)
proving the existence of a unique solution of (\ref{24c}) with $t<
T$ for some $T<\infty$; (b) proving the identification lemma, i.e.,
that the solution of (\ref{24c}) satisfies the conditions of
Proposition \ref{Gpn} and hence is the correlation function of a
unique sub-Poissonian state; (c) constructing the extension of the
solution to all $t>0$ by employing the positive definiteness
obtained in (b).

\subsection{Finite time horizon}

For $\vartheta, \vartheta'\in \mathds{R}$ such that $\vartheta <
\vartheta'$, we set, cf. (\ref{24}),
\begin{gather}
  \label{21a}
T (\vartheta',\vartheta) = \frac{\vartheta' -
\vartheta}{2\alpha}\exp\left( - \langle
\phi \rangle
e^{\vartheta'}\right).
\end{gather}
For a fixed $\vartheta'\in \mathds{R}$, $T (\vartheta', \vartheta))$
can be made as big as one wants by taking small enough $\vartheta$.
However, if $\vartheta$ is fixed, then
\begin{equation}
  \label{N1}
 \sup_{\vartheta' > \vartheta} T(\vartheta', \vartheta) = \frac{\delta (\vartheta)}{2 \alpha} \exp\left( -
\frac{1}{\delta (\vartheta)}\right) =: \tau(\vartheta) < \infty,
\end{equation}
where $\delta(\vartheta)$ is the unique positive solution of the
equation
\begin{equation}
  \label{N2}
\delta e^\delta = \exp\left(- \vartheta - \log \langle
\phi \rangle \right).
\end{equation}
\begin{remark}
  \label{JJrk}
The supremum in (\ref{N1}) is attained at $ \vartheta' = \vartheta +
\delta (\vartheta)$. Note also that $\delta (\vartheta) \to 0$, and
hence $\tau(\vartheta) \to 0$, as $\vartheta \to +\infty$.
\end{remark}
\begin{lemma}
  \label{3tm}
For an arbitrary  $ \vartheta \in \mathds{R}$, the problem in
(\ref{24c}) with $k_{0}\in \mathcal{K}_{\vartheta}$ has a unique
solution $k_{t}\in \mathcal{K}_{\vartheta+ \delta(\vartheta)}$ on
the time interval $[0, \tau(\vartheta))$.
\end{lemma}
\begin{proof}
Take $T <  \tau(\vartheta)$ and then pick $\vartheta' \in
(\vartheta, \vartheta + \delta(\vartheta))$ such that $ T<
T(\vartheta', \vartheta)$. Let $\mathcal{L}(\mathcal{K}_{\vartheta},
\mathcal{K}_{\vartheta'})$ stand for the Banach space of bounded
linear operators acting from $\mathcal{K}_{\vartheta}$ to
$\mathcal{K}_{\vartheta'}$ equipped with the corresponding operator
norm. Our aim is to construct the family
\begin{equation}
  \label{40c}
S_{\vartheta'\vartheta} (t) \in \mathcal{L}(\mathcal{K}_{\vartheta},
\mathcal{K}_{\vartheta'}), \qquad t\in [0, T ( \vartheta',
\vartheta)),
\end{equation}
defined by the series
\begin{equation}
  \label{40b}
S_{\vartheta'\vartheta} (t) = \sum_{n=0}^\infty \frac{t^n}{n!}
\left( L^{ \Delta}\right)^n_{\vartheta'\vartheta}.
\end{equation}
In (\ref{40b}), $\left( L^{\Delta}\right)^0_{\vartheta'\vartheta}$
is the embedding operator and
\begin{equation}
  \label{40d}
\left( L^{\Delta}\right)^n_{\vartheta'\vartheta} := \prod_{l=1}^n
L^{\Delta}_{\vartheta_l \vartheta_{l-1}}, \quad \vartheta_l =
\vartheta + l(\vartheta'- \vartheta)/n,
\end{equation}
for $n\in \mathds{N}$. Now we take into account that $\vartheta_l -
\vartheta_{l-1}= (\vartheta'- \vartheta)/n$ and that $L^{\Delta}$
satisfies (\ref{24}). Then we get
\begin{equation}
  \label{DA}
\|L^{\Delta}_{\vartheta_l \vartheta_{l-1}}\|  \leq
\left(\frac{n}{e}\right)(\vartheta' - \vartheta)\left\{ 2 \alpha
\exp\left( \langle \phi \rangle e^{\vartheta'}\right)\right\}^{-1}\\[.2cm]
 \leq  n \big{/}e T (\vartheta', \vartheta),
\end{equation}
see (\ref{24}) and (\ref{21a}). Next we apply (\ref{DA}) in
(\ref{40d}) and conclude that the series in (\ref{40b}) converges in
the operator norm, uniformly on $[0,T]$, to the operator-valued
function $[0,T] \ni t \mapsto S_{\vartheta'\vartheta} (t) \in
\mathcal{L}(\mathcal{K}_{\vartheta}, \mathcal{K}_{\vartheta'})$ such
that
\begin{equation}
  \label{51}
\forall t\in [0,T]\qquad \|S_{\vartheta'\vartheta} (t) \| \leq
\frac{T (\vartheta', \vartheta)}{T (\vartheta', \vartheta) - t}.
\end{equation}
Likewise, for $\vartheta'' \in (\vartheta' , \vartheta +\delta
(\vartheta)]$, we get
\begin{eqnarray}
  \label{52}
 \frac{d}{dt} S_{\vartheta''\vartheta} (t) & = &
 \sum_{n=0}^\infty \frac{t^n}{n!} \left( L^{
\Delta}\right)^{n+1}_{\vartheta''\vartheta}\\[.2cm] & = & \sum_{n=0}^\infty
\frac{t^n}{n!} L^{\Delta}_{\vartheta''
 \vartheta'} \left( L^{\Delta}\right)^n_{\vartheta'\vartheta} = L^{\Delta}_{\vartheta''
 \vartheta'} S_{\vartheta'\vartheta} (t), \ \quad  t\in [0,T] \nonumber
\end{eqnarray}
Then
\begin{equation}
  \label{53}
 k_{t} = S_{\vartheta'\vartheta} (t) k_{0} \in
\mathcal{K}_{\vartheta'} \subset
\mathcal{D}(L^{\Delta}_{\vartheta''}),
\end{equation}
see Lemma \ref{1lm}, is a solution of (\ref{24c}) on the time
interval $[0, \tau(\vartheta))$ since $T< \tau(\vartheta)$ has been
taken in an arbitrary way.

Let us prove that the solution given in (\ref{53}) is unique. In
view of the linearity, to this end it is enough to show that the
problem in (\ref{24c}) with the zero initial condition has a unique
solution. Assume that $v_t\in \mathcal{D}(L^{\Delta}_{\vartheta +
\delta(\vartheta)})$ is one of the solutions. Then $v_t$ lies in
$\mathcal{K}_{\vartheta''}$ for each $\vartheta'' > \vartheta +
\delta(\vartheta)$, see (\ref{19}). Fix any such $\vartheta''$ and
then take $t < \tau(\vartheta)$ such that $t< T (\vartheta'',
\vartheta+ \delta(\vartheta))$. Then, cf. (\ref{24a}),
\begin{eqnarray*}
  %\label{54}
v_t & = & \int_0^t L^{\Delta}_{\vartheta''
 \bar{\vartheta}} v_s d s \\[.2cm]
& = & \int_0^t \int_0^{t_1} \cdots \int_0^{t_{n-1}} \left(
L^{\Delta}\right)^n_{\vartheta''\bar{\vartheta}} v_{t_n} d t_n
\cdots d t_1, \nonumber
\end{eqnarray*}
where $\bar{\vartheta} := \vartheta + \delta(\vartheta)$ and $n\in
\mathds{N}$ is an arbitrary number. Similarly as above we get from
the latter
\[
\|v_t \|_{\vartheta''} \leq \frac{t^n}{n!} \left(\frac{n}{e T
(\vartheta'',\bar{ \vartheta})}\right)^n \sup_{s\in
[0,t]}\|v_s\|_{\bar{\vartheta}}.
\]
Since $n$ is an arbitrary number, this yields $v_s =0$ for all $s\in
[0,t]$. The extension of this result to all $t <\tau (\vartheta)$
can be done by repeating this procedure due times.
\end{proof}
\begin{remark}
  \label{Decrk}
Similarly as in obtaining (\ref{52}) we have that for  all
$\vartheta_0, \vartheta_1, \vartheta_2 \in \mathds{R}$ such that
$\vartheta_0 < \vartheta_1 < \vartheta_2$, the following holds
\begin{eqnarray}
  \label{D1}
& & \quad S_{\vartheta_2 \vartheta_0} (t+s) = S_{\vartheta_2
\vartheta_1} (t) S{\vartheta_1
\vartheta_0} (s), \\[.2cm] & &  t\in [0, T(\vartheta_2, \vartheta_1)), \quad s \in [0,
T(\vartheta_1, \vartheta_0)). \nonumber
\end{eqnarray}
\end{remark}

\subsection{The identification lemma}

Here we show that the solution of (\ref{24c}) given in (\ref{53})
has the property $k_t \in \mathcal{K}^\star_{\vartheta}$, see
(\ref{19a}). To some extent, we  follow the way of proving Theorem
3.7 in \cite{BKKK}. However, due to an elegant argument provided by
the Denjoy-Carleman theorem \cite{DC}, the present proof is more
complete and transparent.
\begin{lemma}
  \label{Id1lm}
Let $\vartheta^*$ be as in Corollary \ref{Gco}. Then for each $t\in
[0, T(\vartheta, \vartheta^*))$, the operator defined in (\ref{40b})
has the property
\begin{equation}
  \label{63}
  S_{\vartheta \vartheta^*}(t):\mathcal{K}^\star_{\vartheta^*} \to
  \mathcal{K}^\star_{\vartheta}.
\end{equation}
\end{lemma}
\begin{proof}
Let $\mu_0\in \mathcal{P}_{\rm exp} (\Gamma)$ be such that
$k_{\mu_0} \in \mathcal{K}_{\vartheta^*}^\star$, see Proposition
\ref{Gpn}. For $\Lambda\in \mathcal{B}_{\rm b}(\mathds{R}^d)$,  let
$\mu^\Lambda_0$ and $R^\Lambda_{\mu_0}$ be as in (\ref{9c}). For
$N\in \mathds{N}$, we then set
\begin{equation}
  \label{64}
R^{\Lambda,N}_0 (\eta) = R^\Lambda_{\mu_0} (\eta) I_N (\eta), \qquad
\eta \in \Gamma_0,
\end{equation}
where $I_N (\eta)=1$ whenever $|\eta| \leq N$ and
$I_N (\eta)=0$ otherwise. Set
\begin{eqnarray}
  \label{65}
\mathcal{R} & = & L^1 (\Gamma_0, d \lambda) , \quad
\mathcal{R}_\beta = L^1 (\Gamma_0, b_\beta d \lambda), \\[.2cm]
b_\beta (\eta) & := & \exp\bigg{(} \beta |\eta|\bigg{)}, \qquad \beta >0.\nonumber
\end{eqnarray}
Let $\|\cdot\|_{\mathcal{R}}$ and $\|\cdot\|_{\mathcal{R}_\beta}$ be
the norms of the spaces introduced in (\ref{65}) and $\mathcal{R}^+$
and $\mathcal{R}^+_\beta$ be the corresponding cones of positive
elements. For each $\beta>0$, $R^{\Lambda,N}_0$ defined in
(\ref{64}) lies in $\mathcal{R}_\beta^+ \subset \mathcal{R}^+$ and
is such that $\|R^{\Lambda,N}_0\|_{\mathcal{R}} \leq 1$. By means of
perturbative methods developed in \cite{TV}, see \cite[Section
3.2]{BKKK}, it is possible to show that $L^*$ related by (\ref{1b})
to $L$ given in (\ref{LF}) generates the evolution of states $\mu_0
\mapsto \mu_t$, $t\geq 0$, whenever $\mu_0$ has the property
$\mu_0(\Gamma_0)=1$, which is the case for $\mu^\Lambda_0$.
Moreover, for each $t\geq 0$, the mentioned $\mu_t$ is absolutely
continuous with respect to $\lambda$, and the equation for $R_t =
d\mu_t / d \lambda$ corresponding to (\ref{1a}) can be written in
the form
\begin{equation}
  \label{66}
\frac{d}{dt} R_t = L^\dagger R_t, \qquad R_t|_{t=0}=R_{\mu_0},
\end{equation}
where, cf. (\ref{15}), $L^\dagger$ is defined by the relation
$L^\dagger R = d(L^* \mu)/d\lambda$, and hence acts according to the
following formula
\begin{gather}
  \label{66a}
( L^\dagger R)(\eta) =  \sum_{y\in \eta} \int_{\mathds{R}^d} a (x-y)
e(\tau_y;\eta) R(\eta\setminus y
\cup x, \eta) d x -  \Psi(\eta) R(\eta), \nonumber\\[.2cm]
\Psi(\eta) :=  \sum_{x\in \eta} \int_{\mathds{R}^d} a (x-y)
e(\tau_y;\eta) d y .
\end{gather}
Like in the proof of \cite[Theorem 3.7]{BKKK}, one shows that
$L^\dagger$ generates a stochastic $C_0$-semigroup, $S_R:=
\{S_R(t)\}_{t\geq 0}$, on $\mathcal{R}$, which leaves invariant each
$\mathcal{R}_\beta$, $\beta >0$. Then the solution of (\ref{66}) is
$R_t = S_R (t) R_0$. For $R_0^{\Lambda,N}$ as in (\ref{64}), we then
set
\begin{equation}
  \label{67}
R_t^{\Lambda,N} (t) = S_R(t) R_0^{\Lambda,N}, \qquad t>0.
\end{equation}
Then $R_t^{\Lambda,N}\in \mathcal{R}_\beta^+ \subset \mathcal{R}^+$
and $\|R^{\Lambda,N}_t\|_{\mathcal{R}} \leq 1$. This yields that,
for each $G\in B_{\rm bs}^\star (\Gamma_0)$, see (\ref{9f}) and
(\ref{9g}), the following holds
\begin{equation}
  \label{68}
 \langle \! \langle KG , R^{\Lambda,N}_t \rangle\!\rangle \geq 0,
 \qquad t\geq 0.
\end{equation}
The integral in (\ref{68}) exists as $R_t^{\Lambda,N}\in
\mathcal{R}_\beta$ and $KG$ satisfies (\ref{9b}). Moreover, like in
(\ref{24}), for each $\beta'$ such that $0< \beta' < \beta$, we
derive from (\ref{66a}) the following estimate
\begin{equation*}
 % \label{68a}
\|L^\dagger R\|_{\mathcal{R}_{\beta'}} \leq \frac{2\alpha \|
R\|_{\mathcal{R}_{\beta}} }{e(\beta - \beta')}.
\end{equation*}
This allows us to define the corresponding bounded operators
$(L^\dagger)^n_{\beta'\beta} : \mathcal{R}_{\beta} \to
\mathcal{R}_{\beta'}$, $n\in \mathds{N}$, cf. (\ref{40d}), the norms
of which satisfy
\begin{equation}
  \label{68b}
\|(L^\dagger)^n_{\beta'\beta} \| \leq n^n \left(
e\bar{T}(\beta,\beta')\right)^{-n}.
\end{equation}
On the other hand, we have that, cf. (\ref{9d}) and (\ref{64}),
\begin{equation}
  \label{69}
k_0^{\Lambda,N} (\eta) :=  \int_{\Gamma_0} R^{\Lambda,N}_0
(\eta\cup\xi)\lambda(d\xi)
\end{equation}
is such that $k_0^{\Lambda,N}\in \mathcal{K}_{\vartheta^*}^\star$,
and hence we may get
\begin{equation}
  \label{70}
k_t^{\Lambda,N} = S_{\vartheta \vartheta^*}(t) k_0^{\Lambda,N},
\qquad t\in [0,T(\vartheta, \vartheta^*)),
\end{equation}
where $S_{\vartheta \vartheta^*}(t)$ is given in (\ref{40b}).
 Then the proof of (\ref{63}) consists in showing:
\begin{eqnarray}
  \label{71}
 &(i) & \quad \forall G\in B^\star_{\rm bs}(\Gamma_0) \qquad
 \langle \! \langle G, k^{\Lambda,N}_t \rangle \!\rangle \geq
 0;\\[.2cm]
 &(ii)& \quad  \langle \!\langle G, S^1_{\vartheta \vartheta^*} (t) k_0 \rangle \!\rangle =
 \lim_{\Lambda \to \mathds{R}^d } \lim_{N\to +\infty} \langle \! \langle G, k^{\Lambda,N}_t \rangle \!\rangle .\nonumber
\end{eqnarray}
To prove claim {\it (i)} of (\ref{71}) for $G\in B^\star_{\rm
bs}(\Gamma_0)$, cf. (\ref{9g}), we set
\begin{equation}
\label{72} \varphi_G (t) = \langle \! \langle KG , R^{\Lambda,N}_t
\rangle\!\rangle, \quad \ \  \psi_G (t) = \langle \! \langle G ,
k^{\Lambda,N}_t \rangle\!\rangle,
\end{equation}
where $\psi_G$ is defined for $t$ as in (\ref{70}). For a given
$t\in (0,T (\vartheta, \vartheta^*))$, we pick
$\vartheta'<\vartheta$ such that $t < T (\vartheta', \vartheta^*)$,
and hence $k_{s}^{\Lambda,N} \in \mathcal{K}_{\vartheta'}$ for $s\in
[0,t]$. Then the  direct calculation based on (\ref{52}) yields for
the $n$-th derivative
\[
\psi_G^{(n)} (t) = \langle \! \langle G, (L^\Delta)^n_{\vartheta
\vartheta'} k^{\Lambda,N}_t \rangle \!\rangle, \qquad n \in
\mathds{N}.
\]
As in obtaining (\ref{51}) we then get from the latter
\begin{equation}
  \label{73}
|\psi_G^{(n)} (t) | \leq A^n n^n C_{\vartheta'}(G) \sup_{s\in
[0,t]}\|k^{\Lambda,N}_s\|_{\vartheta'}.
\end{equation}
Here $A= 1/ e T(\vartheta, \vartheta')$ and
\[
C_{\vartheta'}(G)  = \int_{\Gamma_0} |G(\eta)|\exp\left( \vartheta '
|\eta|\right) \lambda(d\eta) <\infty,
\]
as $G\in B_{\rm bs}(\Gamma_0)$, see Definition \ref{Gdef}. Likewise,
from  (\ref{67}) we get
\[
\varphi^{(n)}_G (t) = \langle \! \langle KG,
(L^\dagger)^n_{\beta'\beta} R^{\Lambda,N}_t \rangle \! \rangle
\]
For the same $t$ as in (\ref{73}), by (\ref{68b}) we have from the
latter
\begin{equation}
  \label{73a}
|\varphi_G^{(n)} (t) | \leq \bar{A}^n n^n C_{\beta'}(G) \sup_{s\in
[0,t]}\|R^{\Lambda,N}_s\|_{\beta'}.
\end{equation}
Here $\bar{A}= 1/ e \bar{T}(\beta', \beta)$ and
\[
C_{\beta'}(G)  =  \esssup_{\eta \in \Gamma_0} |KG(\eta)| \exp
\left(-\beta' |\eta| \right) <\infty,
\]
which holds in view of (\ref{9b}). By (\ref{15}) and (\ref{69}) it
follows that
\[
(L^\Delta k^{\Lambda,N}_0)(\eta)=\int_{\Gamma^2_0}(L^\dagger
R_0^{\Lambda,N})(\eta \cup \xi) \lambda (d\xi),
\]
which then yields
\begin{equation}
  \label{74}
\forall n\in \mathds{N}_0 \qquad \varphi^{(n)}_G(0) =
\psi^{(n)}_G(0).
\end{equation}
By (\ref{73a}) and (\ref{73}) both functions defined in (\ref{72})
are quasi-analytic on $[0,t]$. Then by the Denjoy-Carleman theorem
\cite{DC}, (\ref{74}) implies, see (\ref{68}),
\begin{equation}
  \label{74a}
 \forall t\in [0,T(\vartheta, \vartheta^*)) \qquad \psi_G(t) = \varphi_G(t) \geq 0,
\end{equation}
which yields the first  line in (\ref{71}).
 The convergence in claim
{\it(ii)} of (\ref{71}) is proved in a standard way, see Appendix in
\cite{BKKK}.
\end{proof}
Note that (\ref{74a}) yields also that
\begin{equation}
  \label{74b}
 \forall t\in [0,T(\vartheta, \vartheta^*)) \qquad \langle
 \!\langle G,  q^{\Lambda,N}_t\rangle \!\rangle
= \langle
 \!\langle G,  k^{\Lambda,N}_t\rangle \!\rangle,
\end{equation}
where $G$ and  $k^{\Lambda,N}_t$ are as in (\ref{72}) and
\begin{equation}
  \label{74c}
q^{\Lambda,N}_t (\eta) := \int_{\Gamma_0^2} R^{\Lambda,N}_t (\eta
\cup \xi) \lambda (d\xi).
\end{equation}

\subsection{An auxiliary evolution}

The evolution which we construct now  will be used to extending the
solution $k_t$ given in (\ref{53}) to the global solution as stated
in Theorem \ref{1tm}. The construction employs the operator
\begin{equation}
  \label{55}
(\bar{L} k) (\eta)  =  \sum_{y\in \eta}
\int_{\mathds{R}^d} a (x-y) k(\eta \setminus y \cup x) d
x
\end{equation}
obtained from $L^\Delta$ given in (\ref{15}) by putting $\phi =0$, and then dropping the second term. Hence, like
in (\ref{24}) we get
\begin{equation}
  \label{56}
\| \bar{L}k\|_\vartheta \leq  \frac{2 \alpha \|k\|_{\vartheta''}}{
e(\vartheta - \vartheta'')},
\end{equation}
which allows us to introduce the operators $(\bar{L}_\vartheta,
\mathcal{D} (\bar{L}_\vartheta))$ and
$\bar{L}_{\vartheta\vartheta''}\in
\mathcal{L}(\mathcal{K}_{\vartheta''}, \mathcal{K}_{\vartheta})$
such that, cf. (\ref{24a}),
\begin{equation*}
  %\label{57}
\forall k \in \mathcal{\vartheta''} \qquad
\bar{L}_{\vartheta\vartheta''}k = \bar{L}_{\vartheta} k, \qquad
\vartheta'' < \vartheta.
\end{equation*}
Like above, we have that
\[
\mathcal{K}_{\vartheta''} \subset \mathcal{D}(\bar{L}_{\vartheta}):=
\{ k \in \mathcal{K}_{\vartheta} : \bar{L} k \in
\mathcal{K}_{\vartheta}\}, \qquad \vartheta'' < \vartheta  .
\]
Note that
\begin{equation}
  \label{58}
\bar{L}_{\vartheta\vartheta''} :\mathcal{K}_{\vartheta''}^+ \to
\mathcal{K}_{\vartheta}^+ , \qquad \vartheta'' < \vartheta,
\end{equation}
see (\ref{19b}).  For $n\in \mathds{N}$, we define
$(\bar{L})^n_{\vartheta'\vartheta}$ similarly as in (\ref{40d}) and
denote, cf. (\ref{21a}),
\begin{equation}
  \label{59}
\bar{T}(\vartheta', \vartheta) = (\vartheta' - \vartheta)/2 \alpha,
\qquad \vartheta < \vartheta' .
\end{equation}
Our aim is to study the operator valued function defined by the
series
\begin{equation}
  \label{61}
\bar{S}_{\vartheta'\vartheta} (t) = \sum_{n=0}^\infty \frac{t^n}{n!}
\left(\bar{L}\right)^n_{\vartheta'\vartheta}.
\end{equation}
\begin{lemma}
  \label{A1lm}
For each $\vartheta_0, \vartheta\in \mathds{R}$ such that
$\vartheta_0 < \vartheta$, the series in (\ref{61}) defines a
continuous function
\begin{equation}
  \label{DD}
[0, \bar{T}(\vartheta, \vartheta_0) ) \ni t \mapsto
\bar{S}_{\vartheta\vartheta_0} (t) \in \mathcal{L}(
\mathcal{K}_{\vartheta_0}, \mathcal{K}_{\vartheta}),
\end{equation}
which has the following properties:
\begin{itemize}
  \item[{\it (a)}] For $t$ as in (\ref{DD}), let $\vartheta''\in (\vartheta_0, \vartheta)$ be such that
 $t< \bar{T}(\vartheta'',
\vartheta_0)$. Then, cf. (\ref{52}),
\begin{equation}
  \label{D2}
\frac{d}{dt}\bar{S}_{\vartheta\vartheta_0}(t) = \bar{L}_{\vartheta
\vartheta''} \bar{S}_{\vartheta''\vartheta_0} (t).
\end{equation}
  \item[{\it (b)}]
The problem
\begin{equation}
  \label{60}
 \frac{d}{dt} u_t = \bar{L}_\vartheta u_t , \qquad u_t|_{t=0} = u_0
 \in \mathcal{K}^+_{\vartheta_0},
\end{equation}
has a unique solution $u_t \in \mathcal{K}^+_{\vartheta}$ on the
time interval $[0, \bar{T}(\vartheta, \vartheta_0))$ given by
\begin{equation}
  \label{62}
  u_t = \bar{S}_{\vartheta''\vartheta_0} (t)u_0,
\end{equation}
where, for a fixed $t \in [0, \bar{T}(\vartheta, \vartheta_0))$,
$\vartheta''$ is chosen to satisfy $t< \bar{T}(\vartheta'',
\vartheta_0)$.
\end{itemize}
\end{lemma}
\begin{proof}
Proceeding as in the proof of Lemma \ref{3tm}, by means of the
estimate in (\ref{56}) we prove the convergence of the series in
(\ref{61}). This allows also for proving (\ref{D2}), which yields
the existence of the solution of (\ref{60}) in the form given in
(\ref{62}). The uniqueness is proved analogously as in the case of
Lemma \ref{3tm}. The stated positivity of $u_t$ follows from
(\ref{61}) and (\ref{58}).
\end{proof}
\begin{corollary}
  \label{D1co}
For a given $C>0$, we let in (\ref{60}) and (\ref{62})  $\vartheta_0
=\log C$ and $u_0 (\eta) = C^{|\eta|}$. Then the unique
solution of (\ref{60}) is
\begin{equation}
  \label{85}
u_t (\eta) = C^{|\eta|} \exp\left\{ t(\alpha |\eta|)\right\}.
\end{equation}
This solution can naturally be continued to all $t>0$ for which it
lies in $\mathcal{K}_{\vartheta(t)}$ with
\begin{equation}
  \label{De}
\vartheta(t) =  \log C + t \alpha.
\end{equation}
\end{corollary}
\begin{proof}
In view of the lack of interaction in (\ref{55}), the equations for
particular $u^{(n)}_t$ take the following form
\begin{eqnarray*}
 & &\frac{d}{dt} u_t^{(n)} (x_1, \dots ,x_{n_0})  =\\
 & &\sum_{i=1}^{n} \int_{\mathds{R}^d} a (x- x_i) u_t^{(n)} (x_1,
 \dots , x_{i-1}, x , x_{i+1},\dots,x_n) d x
\quad  n\in  \mathds{N},
\end{eqnarray*}
which for the initial translation invariant $u_0$ yields (\ref{85}).
\end{proof}

\subsection{The global solution}

As follows from Lemmas \ref{3tm} and \ref{Id1lm}, the unique
solution of the problem (\ref{24c}) with $k_{0}\in
\mathcal{K}^\star_{\vartheta^*}$ lies in
$\mathcal{K}_\vartheta^\star$ for $t \in (0, T (\vartheta,
\vartheta^*))$. At the same time, for fixed $\vartheta^*$,  $T
(\vartheta, \vartheta^* )$ is bounded, see (\ref{N1}). This means
that the mentioned solution cannot be directly continued as stated
in Theorem \ref{1tm}. In this subsection,  by a comparison method we
prove that, for $t \in (0, T (\vartheta, \vartheta^*))$, $k_t$
satisfies (\ref{24d}) which is then used to get the continuation in
question, cf. Corollary \ref{D1co}. Recall that the operator $Q_y$,
was introduced in (\ref{13}) and the cone $\mathcal{K}^+_\vartheta$
was defined in (\ref{19b}).
\begin{lemma}
  \label{J1lm}
For each $k_0\in \mathcal{K}_{\vartheta^*}^\star$ and $t \in (0, T
(\vartheta, \vartheta^*))$, $k_t := S_{\vartheta \vartheta^*}(t)k_0$
has the property
\begin{equation*}
  %\label{75}
\left[k_t - e(\tau_y;\cdot) (Q_y k_t)\right] \in
\mathcal{K}_\vartheta^+,
\end{equation*}
holding for Lebesgue-almost all $y\in \mathds{R}^d$.
\end{lemma}
\begin{proof}
 For a fixed $y$, we denote
\[
v_{t,1} = k_t - Q_y k_t, \quad v_{t,2} =  [1- e(\tau_y;\cdot)]
Q_y k_t.
\]
The proof will be done if we show that, for all $G\in B_{\rm
bs}(\Gamma_0)$ such that $G(\eta) \geq 0$ for $\lambda$-almost all
$\eta\in \Gamma_0$, the following holds
\begin{equation}
  \label{76}
\langle \! \langle G, v_{t,j} \rangle \! \rangle \geq 0, \qquad
j=1,2.
\end{equation}
Let $\Lambda$, $N$, and $k_0^{\Lambda,N}$ be as in (\ref{69}), and
then $k_t^{\Lambda,N}$ be as in (\ref{70}). Next, let
$v_{t,j}^{\Lambda,N}$, $j=1,2$, be defined as above with $k_t$
replaced by $k_t^{\Lambda,N}$. By (\ref{74b}) and (\ref{74c}) we
then get
\begin{eqnarray}
  \label{78}
\langle \! \langle G, Q_y k^{\Lambda,N}_t \rangle \! \rangle & = &
\int_{\Gamma_0} \widetilde{G}(\eta) k^{\Lambda,N}_t (\eta)
\lambda(d \eta) \\[.2cm]
& = & \int_{\Gamma_0} \int_{\Gamma_0} \widetilde{G}(\eta)
R^{\Lambda,N}_t (\eta \cup \xi) \lambda (d\eta) \lambda (d\xi),
\nonumber
\end{eqnarray}
where
\begin{equation*}
  %\label{78a}
  \widetilde{G}(\eta):=\sum_{\xi\subset \eta} e(t_y; \xi) G(\eta\setminus
\xi).
\end{equation*}
Furthermore, by (\ref{78}) we get
\begin{eqnarray}
  \label{78b}
& & \langle \! \langle G, Q_y k^{\Lambda,N}_t \rangle \! \rangle \\[.2cm] & & \quad =
\int_{\Gamma_0} G(\eta)  \int_{\Gamma_0} \left( \int_{\Gamma_0}
e(t_y;\zeta) R^{\Lambda,N}_t(\eta\cup \xi \cup \zeta) \lambda ( d
\zeta) \right)
\lambda(d\eta) \lambda ( d \xi) \qquad \nonumber \\[.2cm]
& & \quad =  \int_{\Gamma_0} G(\eta) \int_{\Gamma_0}
\left( \sum_{\zeta\subset \xi} e(t_y;\zeta)\right)
R^{\Lambda,N}_t(\eta \cup \xi) \lambda(d\eta)
\lambda ( d \xi).\nonumber \qquad
\end{eqnarray}
By (\ref{12}) we have that
\begin{equation*}
 % \label{77}
 \sum_{\zeta \subset \xi} e(t_y;\zeta) = e(\tau_y;\xi).
\end{equation*}
We apply this in the last line of (\ref{78b}) and obtain
\begin{eqnarray}
  \label{78c}
& & \langle \! \langle G, Q_y k^{\Lambda,N}_t \rangle \! \rangle
\\[.2cm] & & \quad =  \int_{\Gamma_0} G(\eta) \int_{\Gamma_0}
 e(\tau_y;\xi)
R^{\Lambda,N}_t(\eta\cup \xi) \lambda(d\eta)
\lambda ( d \xi)\nonumber \qquad \\[.2cm]
& & \quad \leq \int_{\Gamma_0} G(\eta) \int_{\Gamma_0}
 R^{\Lambda,N}_t(\eta\cup \xi) \lambda(d\eta)
\lambda ( d \xi)\nonumber \qquad \\[.2cm] & & \quad = \langle \! \langle G,  k^{\Lambda,N}_t \rangle \!
\rangle, \nonumber
\end{eqnarray}
which after the limiting transition as in (\ref{71}) yields
(\ref{76}) for $j=1$. For the same $G$, we set $\bar{G} =
e(\tau_y;\cdot) G$. Then by (\ref{12}) and the second line in
(\ref{78c}) we get
\begin{equation*}
  %\label{79}
\langle \! \langle \bar{G}, Q_y k^{\Lambda,N}_t \rangle \! \rangle
\leq \langle \! \langle G, Q_y k^{\Lambda,N}_t \rangle \! \rangle,
\end{equation*}
which after the limiting transition as in (\ref{71}) yields
(\ref{76}) for $j=2$.
\end{proof}
\begin{lemma}
  \label{J2lm}
Let $C>0$ be such that the initial condition in (\ref{24c})
satisfies $k_{\mu_0}(\eta) =k_0 (\eta) \leq C^{|\eta|}$.
Then for all $t< T (\vartheta, \vartheta^*)$ with $\vartheta^* =
\log C$ and any $\vartheta> \vartheta^*$, the unique solution of
(\ref{24c}) given by the formula
\begin{equation}
  \label{Sol}
k_t = S_{\vartheta \vartheta^*} (t) k_0
\end{equation}
satisfies (\ref{24d}) for $\lambda$-almost all $\eta\in \Gamma_0$.
\end{lemma}
\begin{proof}
Take any $\vartheta > \vartheta^*$ and fix $t< T (\vartheta,
\vartheta^*)$; then pick $\vartheta^1 \in (\vartheta^* , \vartheta)$
such that $t< T (\vartheta^1, \vartheta^*)$. Next take $\vartheta^2,
\vartheta^3\in \mathds{R}$ such that $\vartheta^1<\vartheta^2 <
\vartheta^3$ and $t< \bar{T} (\vartheta^3, \vartheta^2)$. The latter
is possible since $\bar{T}$ depends only on the difference
$\vartheta_3 - \vartheta_2$, see (\ref{59}). For the fixed $t$, $k_t
\in \mathcal{K}_{\vartheta^1}^\star \hookrightarrow
\mathcal{K}_{\vartheta^3}^\star$, and hence one can write
\begin{eqnarray}
 \label{83}
 u_t & = &
\bar{S}_{\vartheta^3 \vartheta^*}(t) u_0 \\[.2cm]
 & = & (u_0 - k_0)+ k_t + \int_0^t \bar{S}_{\vartheta^3 \vartheta^2} (t-s)
 D_{\vartheta^2\vartheta^1} k_s ds, \nonumber
\end{eqnarray}
where
\begin{equation*}
  %\label{81}
  D_{\vartheta \vartheta''} = \bar{L}_{\vartheta \vartheta''} -
L^\Delta_{\vartheta \vartheta''}, \qquad    D_\vartheta =
\bar{L}_\vartheta - L^\Delta_\vartheta,
\end{equation*}
and the latter two operators are as in (\ref{60}) and (\ref{24c})
respectively. By Lemma \ref{Id1lm}, for $s\leq t$, $k_s \in
\mathcal{K}_{\vartheta^1}^\star$. By (\ref{15}), (\ref{55}), and
Lemma \ref{J1lm} we have that $D_{\vartheta^2\vartheta^1}:
\mathcal{K}_{\vartheta^1}^\star \to \mathcal{K}_{\vartheta^2}^+$.
Then by Lemma \ref{A1lm} the third summand in the second line in
(\ref{83}) is in $\mathcal{K}_{\vartheta^3}^+$ which completes the
proof since $u_0 - k_0$ is also positive.
\end{proof}
\vskip.1cm \noindent {\it Proof of Theorem \ref{1tm}.} According to
Definition \ref{S1df} and Remark \ref{D1rk} the map $[0,+\infty)\ni
t \mapsto k_t \in \mathcal{K}^\star$ is the solution in question if:
(a) $k_t(\emptyset)=1$; (b) for each $t>0$, there exists
$\vartheta''\in \mathds{R}$ such that $k_t \in
\mathcal{K}_{\vartheta''}$ and  $\frac{d}{dt} k_t =
L^\Delta_\vartheta k_t$ for each $\vartheta> \vartheta''$.

Let $k_0$ and $C>0$ be as in the statement of Theorem \ref{1tm}. Set
$\vartheta^*= \log C$. Then, for $\vartheta = \vartheta^* +
\delta(\vartheta^*)$, see (\ref{N1}) and (\ref{N2}), $k_t$ as given
in (\ref{Sol}) is a unique solution of (\ref{24c}) in
$\mathcal{K}_\vartheta$ on the time interval $[0,T(\vartheta,
\vartheta^*))$. By (\ref{15}) we have
\[
\left(\frac{d}{dt} k_t\right)(\emptyset) = (L^\Delta
k_t)(\emptyset) =0,
\]
which yields that $k_t(\emptyset) = k_0(\emptyset) =1$. By Lemma \ref{Id1lm} $k_t \in
\mathcal{K}_\vartheta^\star$, and hence $k_t$ is the solution in
question for $t< \tau( \vartheta^*)$.
 According to Lemma \ref{J2lm} $k_t$ lies in
$\mathcal{K}_{\vartheta(t)}$ with $\vartheta (t)$ given in
(\ref{De}). Fix any $\epsilon \in (0,1)$ and then set $s_0=0$, $s_1
=(1-\epsilon) \tau(\vartheta^*)$, and $\vartheta^{*}_1 = \vartheta
(s_1)$. Thereafter, set  $\vartheta^1 = \vartheta^{*}_{1} +
\delta(\vartheta^{*}_{1})$ and
\begin{equation*}
  %\label{82a}
 k_{t+  s_1} = S_{\vartheta^1 \vartheta^{*}_{1}}(t)k_{s_1}, \qquad t \in [0, \tau( \vartheta^{*}_{1})).
\end{equation*}
Note that for $t$ such that $t+s_1 < \tau(\vartheta^*)$,
\[
 k_{t+  s_1} = S_{\vartheta^1 \vartheta^{*}}(t+s_1)k_{0},
\]
see (\ref{D1}). Thus, by Lemmas \ref{Id1lm} and \ref{J2lm} the map
$[0, s_1 + \tau(\vartheta^{*}_{1})) \ni t \mapsto k_t \in
\mathcal{K}_{\vartheta(t)}$ with
\begin{equation*}
 % \label{82b}
k_t = \left\{ \begin{array}{ll} S_{\vartheta^{*}_1 \vartheta^{*}}(t)
k_0
\quad &t\leq s_1;\\[.3cm] S_{\vartheta^1 \vartheta^{*}_{1}}(t-s_1)
k_{s_1} \quad &t\in [s_1, s_1+ \tau( \vartheta^{*}_{1}))
\end{array} \right.
\end{equation*}
is the solution in question on the indicated time interval. We
continue this procedure by setting $s_n =(1-\epsilon)
\tau(\vartheta^*_{n-1})$, $n\geq 2$, and then
\begin{equation}
  \label{N4}
\vartheta^*_n = \vartheta(s_1+\cdots +s_n ), \qquad \vartheta^n =
\vartheta^*_n + \delta(\vartheta^*_n).
\end{equation}
This yields the solution in question on the time interval $[0, s_1 +
\cdots + s_{n+1}]$ which for $t \in [s_1 +\cdots + s_l, s_1 +\cdots
+ s_{l+1}]$, $l=0, \dots , n$, is given by
\[
k_t =  S_{\vartheta^l \vartheta^*_{l}}(t-(s_1 +\cdots +s_l))
k_{s_l}.
\]
Then the global solution in question exists whenever the series
\[
\sum_{n\geq 1}s_n = (1-\epsilon) \sum_{n\geq 1} \tau(\vartheta^*_n)
\]
diverges. Assume that this is not the case. Then by (\ref{De}) and
(\ref{N4}) we get that both (a) and (b) ought to be true, where (a)
$\sup_{n\geq 1} \vartheta^*_n =: \bar{\vartheta}<+ \infty$ and (b)
$\tau(\vartheta^*_n) \to 0$ as $n\to +\infty$. However, by
(\ref{N1}) and (\ref{N2}) it follows that (a) implies
$\tau(\vartheta^*_n) \geq \tau(\bar{\vartheta}) > 0$, which
contradicts (b). \hskip2.5cm $\square$

\section*{Acknowledgment}
The present research was supported by the European Commission under
the project STREVCOMS PIRSES-2013-612669.

\end{document}